\documentclass[11pt]{article}
\usepackage[top=1in, bottom=1in, left=1in, right=1in]{geometry}
\usepackage{tikz}
\usetikzlibrary{shapes,arrows,matrix,patterns,positioning}
\usepackage{color}
\usepackage{graphicx}
\usepackage{subcaption}
\usepackage{algorithmic}
\usepackage{amssymb}
\usepackage{epstopdf}
\usepackage[ruled]{algorithm2e}
\usepackage{float}
\usepackage{amsmath,amsthm,bm,color,epsfig,enumerate,caption}
\usepackage{hyperref}
\usepackage{booktabs}
\usepackage{multirow}
\usepackage{array}
\usepackage{bigstrut}
\usepackage{mathabx}
\usepackage{physics}

\newtheorem{theorem}{Theorem}[section]

\newtheorem{corollary}[theorem]{Corollary}

\newtheorem{assumption}[theorem]{Assumption}

\newtheorem{remark}[theorem]{Remark}

\DeclareMathOperator*{\argmin}{arg\,min}

\newcommand{\bbN}{\mathbb{N}}

\newcommand{\tb}{\widetilde{b}}
\newcommand{\tf}{\widetilde{f}}
\newcommand{\tA}{\widetilde{A}}
\newcommand{\tM}{\widetilde{M}}
\newcommand{\hk}{\widehat{k}}
\renewcommand{\hm}{\widehat{m}}
\newcommand{\hv}{\widehat{v}}
\newcommand{\hlambda}{\widehat{\lambda}}
\newcommand{\hdelta}{\widehat{\delta}}

\newcommand{\cond}[1]{\kappa\left( #1 \right)}
\newcommand{\condp}[2]{\kappa_{#1}\bigl( #2 \bigr)}
\newcommand{\bimin}[2]{\min \left( #1, #2 \right)}
\newcommand{\bimax}[2]{\max \left( #1, #2 \right)}
\newcommand{\expect}[1]{\mathbb{E}\left[ #1 \right]}
\newcommand{\varance}[1]{\mathbb{V}\left[ #1 \right]}
\newcommand{\diag}[1]{\text{diag}\left( #1 \right)}

\renewcommand{\grad}{\nabla}

\newcommand{\MC}{{M}onte {C}arlo}

\makeatletter
\newcommand*{\extendadd}{
  \mathbin{
    \mathpalette\extend@add{}
  }
}
\newcommand*{\extend@add}[2]{
  \ooalign{
    $\m@th#1\leftrightarrow$%
    \vphantom{$\m@th#1\updownarrow$}
    \cr
    \hfil$\m@th#1\updownarrow$\hfil
  }
}
\makeatother

\begin{document}

\title{Bold Diagrammatic \MC{} in the Lens of Stochastic Iterative
Methods}

\author{Yingzhou Li$^{\,\sharp}$,
        Jianfeng Lu$^{\,\sharp\,\dagger}$
  \vspace{0.1in}\\
  $\sharp$ Department of Mathematics, Duke University\\
  $\dagger$ Department of Chemistry and Department of Physics,
  Duke University\\
}

\maketitle

\begin{abstract} 
  This work aims at understanding of bold diagrammatic Monte Carlo
  (BDMC) methods for stochastic summation of Feynman diagrams from the
  angle of stochastic iterative methods. The convergence enhancement
  trick of the BDMC is investigated from the analysis of condition
  number and convergence of the stochastic iterative
  methods. Numerical experiments are carried out for model systems to
  compare the BDMC with related stochastic iterative approaches.
\end{abstract}


\section{Introduction}
\label{sec:Intro}

Bold(-line) diagrammatic \MC{}~(BDMC) method~\cite{Prokofev2007}
employs bold-line trick in the diagrammatic \MC{}~(DMC) method to
simulate integrands represented by a diagrammatic structure. Such a
method adopts mathematical tools including \MC{} sampling of the
diagram and iterative method for the bold-line trick. This note first
establishes a solid mathematical understanding of the iterative method
proposed in the original BDMC paper~\cite{Prokofev2007}. Second, this
note clarifies the relationship between the iterative method in BDMC
and stochastic iterative methods. Based on the explicit connection, a
few stochastic iterative methods~\cite{Polyak1964, Tan2016, Duchi2011,
  Kingma2015, Bottou2016}, widely used and extensively tested in the
field of machine learning, are reintroduced in this note as potential
alternatives to BDMC with potentially faster convergence.

Both DMC and BDMC are proposed for ``many-electron problem'' which
involves interacting electrons. In order to describe an interacting
electron system, the dimension of the Hilbert space grows
exponentially in the system size; the high dimensionality becomes a
fundamental difficulty for numerical treatment. The quantum \MC{}
methods are thus natural candidates for these problems. Conventional
quantum Monte Carlo methods calculate solutions on finite-size
lattices, and then estimate the solution of the thermodynamic limit
(thus infinite system) via extrapolations, see e.g. reviews
\cite{Foulkes2001, Ceperley2010, Kolorenc2011, Austin2012}. On the
other hand, the DMC and BDMC sample and sum the truncated Feynman
diagram of the infinite system \cite{VanHoucke2008}. The Feynman
diagram is well developed and widely used tools in many-body
perturbation theory, see e.g. books \cite{Mattuck1992,
  FetterWalecka2003}. In particular, the summation of series of
Feynman diagrams works well for those that are convergent and sign
positive. In order to obtain the summation of the infinite long
diagram, the extrapolation technique is applied to a few results
corresponding to different numbers of truncation orders. However, for
many systems, the series of diagrams are asymptotic~(e.g., for strong
coupling systems) and sign-alternating. No solution, so far, fully
addresses these issues. Techniques have been developed to enlarge the
radius of the convergence and reduce the number of terms in the
diagram. BDMC is one of the promising technique among those.  BDMC,
instead of summing diagrams directly, sums all the bold-line diagrams
for irreducible single-particle self-energy $\Sigma$ and pair
self-energy $\Pi$ following {Dyson} and {Bethe-Salpeter} equation
respectively~\cite{Prokofev2007,Prokofev2008}. Based on the
``sign-blessing'' phenomenon, BDMC was successfully applied to
one-particle $s$-scattering problem~\cite{Prokofev2007}, the BCS--BEC
({Bardeen}-{Cooper}-{Schrieffer}-{Bose}-{Einstein}-Condensation)
crossover in the strongly imbalanced
regime~\cite{Prokofev2008,Prokofev2008a}, unitary Fermi
gas~\cite{VanHoucke2012}, Fermionized frustrated
spins~\cite{Kulagin2013}, two-dimensional {Hubbard}
model~\cite{LeBlanc2015}, etc.

\begin{figure}[htp]
    \centering
    \includegraphics[height=2in]{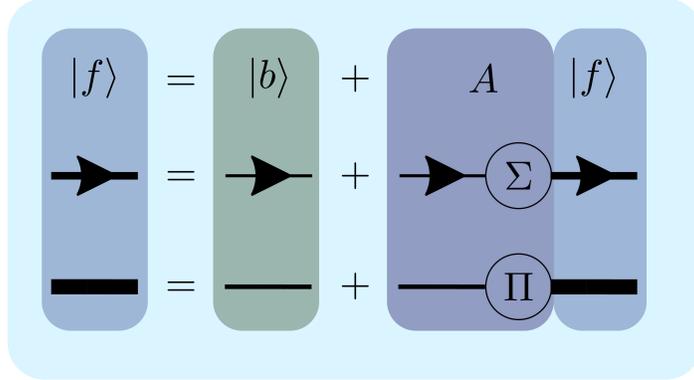}
    \caption{Relationship between linear system, {Dyson} equation and
    {Bethe-Salpeter} equation.}
    \label{fig:FeymanDiag}
\end{figure}

As from the original paper \cite{Prokofev2007}, BDMC can be viewed as trying
to solve a self-consistent linear equation:
\begin{equation} \label{eq:BDMC-self-consistent}
    \ket{f} = \ket{b} + A \ket{f},
\end{equation}
where $\ket{f}$ is an unknown vector, $\ket{b}$ is a given vector, and
$A$ is a linear operator.  Figure~\ref{fig:FeymanDiag} provides the
connection between \eqref{eq:BDMC-self-consistent} with either the
{Dyson} equation or the {Bethe}-{Salpeter} equation
\cite{VanHoucke2012}. Since either $\Sigma$ or $\Pi$ involves infinite
terms of diagrams, the evaluation is carried out via a stochastic
procedure up to a given number of terms. The evaluation of $A\ket{f}$,
therefore, is stochastic, where the error is controlled by the number
of \MC{} sampling. {Prokof'ev} and {Svistunov} in 2007 proposed an
iterative method to solve for $\ket{f}$ in
\eqref{eq:BDMC-self-consistent} under the stochastic setting, whose
connection to conventional iterative algorithm in numerical linear
algebra is not obvious from the first sight. As {BDMC} achieves
success in many interacting systems and shows great promise,
establishing a concrete understanding of the iterative method in a
mathematical way is crucial for further improvement of the method, and potentially adapt the method to other applications.

In this note, we interpret the ``magic'' method proposed
in~\cite{Prokofev2007} as a combination of two crucial steps. The
first step replaces the original operator $A$ by a quadratic
polynomial of $A$, $p(A)$, such that $p(0)=0$, $p(A) \succeq 0$ and
potentially $\cond{p(A)} \ll \cond{A^*A}$, where ``$A \succeq 0$''
means that $A$ is a positive semidefinite matrix and $\cond{A}$
denotes the condition number of matrix $A$.  Here, ``$p(0)=0$''
guarantees the equality in \eqref{eq:BDMC-self-consistent},
``$p(A) \succeq 0$'' guarantees the convergence of the iterative
method, and ``$\cond{p(A)} \ll \cond{A^*A}$'' enables faster
convergence rate. Based on this understanding, we suggest another form
of the quadratic polynomial such that the similar properties can be
achieved for a wider range of $A$. In the second step, a fixed point
iteration with adaptive stepsize is applied to
\eqref{eq:BDMC-self-consistent} with $A$ being replaced by $p(A)$, and
the corresponding update on $\ket{b}$.  When $A$ is a {Hermitian}
matrix, the second step can be viewed as a method of stochastic
gradient descent. Hence, later in the note, we employ stochastic
gradient descent methods from machine learning as alternative methods.
All methods are tested on synthetic stochastic matrix $A$ instead of
diagrams for real physical systems; those will be considered for
future works.

In this note, we will provide mathematical understanding of the
stochastic iterative method~\cite{Prokofev2007} in
Section~\ref{sec:BDMC}. Section~\ref{sec:SGDs} lists several
alternative stochastic iterative methods. All the mentioned methods
are tested and compared in Section~\ref{sec:NumRes}. Finally, in
Section~\ref{sec:Conclusion}, we conclude the note together with
discussion of possible future works.


\section{Numerical method of {BDMC}}
\label{sec:BDMC}

Recall the fixed-point problem BDMC tries to solve:
\begin{equation} \label{eq:BDMC-fixed-point}
    \ket{f} = \ket{b} + A \ket{f}.
\end{equation}
In the viewpoint of linear algebra, we rewrite the equation as
\begin{equation} \label{eq:BDMC}
    M \ket{f} = \ket{b},
\end{equation}
where $M = I - A$, $I$ is the identity matrix of the same size as $A$.
{BDMC} proposes replacements
$\ket{b} \rightarrow \ket*{\tb} = \ket{b} - \lambda A \ket{b}$ and
$A \rightarrow \tA = (1+\lambda) A - \lambda A^2$ to ensure the
convergence of the iterative method, where $\lambda$ is a constant
related to the spectrum of $A$. Then a simple fixed-point iteration is
coupled with a special {N{\o}rlund} means to solve
\eqref{eq:BDMC-fixed-point} with $\tA$ and $\ket*{\tb}$. In the
following subsections, we reinterpret the former as a preconditioning
step and the latter as a stochastic gradient descent method with
diminishing stepsize.

In the rest of this note, we would stick to linear algebra notations as
in \eqref{eq:BDMC}. Accordingly, we have $\ket*{\tb} =
(1-\lambda)\ket{b} + \lambda M \ket{b}$, $\tM = I-\tA = (1-\lambda) M +
\lambda M^2$. Additionally, we follow the assumption as in
\cite{Prokofev2007} that $A$ is Hermitian, i.e., $A^*=A$. Therefore,
both $M$ and $\tM$ are Hermitian as well. 

\subsection{Preconditioning indefinite matrices}
\label{sec:Precond}

For almost all first-order iterative methods, positivity of $M$ is
required for convergence. Methods that work for indefinite matrices,
such as MINRES~\cite{Paige1975}, GMRES~\cite{Saad1986}, adopt some
transforms of $M$, e.g., $M^*M$, $M^2$, to turn the matrix in the
iterative method positive definite. Another important property of $M$ or
$\tM$ related to convergence rate is the condition number. In general,
smaller condition number leads to faster convergence. However, treatment
as $M^*M$ or $M^2$ squares the condition number which is undesirable in
practice. In this section, we analyze the positivity of $\tM$ and its
condition number comparing to $\cond{M^2}$.

Assume $M$ is an indefinite invertible matrix of size $n$ by $n$.
According to the earlier assumption, $M$ is Hermitian. Let
$M = Q \Lambda Q^*$ be the eigenvalue decomposition of $M$, where $Q$
is a unitary matrix of size $n$ by $n$, $\Lambda$ is a diagonal matrix
with $M$'s eigenvalues, $\{m_1,m_2,\dots,m_n\}$, in decreasing order,
i.e.,
$m_1 \geq m_2 \geq \cdots \geq m_\ell > 0 > m_{\ell+1} \geq \cdots
\geq m_n$
for $1 < \ell < n$. To simplify the presentation in the sequel, we
introduce handy notations as, $L_+ = \max_{1 \leq i \leq n} m_i$,
$L_- = \min_{1 \leq i \leq n} m_i$, $\tau_+ = \min_{m_i > 0} m_i$, and
$\tau_- = \max_{m_i < 0} m_i$, which define the boundaries of the
positive and negative spectrum of $M$.

$\tM = (1-\lambda)M + \lambda M^2$ inherits the same
eigenvectors as $M$. The eigenvalues of $\tM$ are $\left\{(1-\lambda)m_i +
\lambda m_i^2\right\}_{i=1}^n$. Denote the quadratic polynomial
depending on parameter $\lambda$ as $p_\lambda(x) = \lambda x^2 +
(1-\lambda)x$. The eigenvalues of $\tM$, therefore, are polynomial
$p_\lambda(x)$ acting on the eigenvalues of $M$. $\tM$ being a positive
definite matrix is equivalent to $p_\lambda(m_i) > 0$ for all
$m_1, \dots, m_n$. Since $p_\lambda(x)$ is a quadratic polynomial with
zero being one of its root, $p_\lambda(\tau_-) > 0$ and
$p_\lambda(\tau_+) > 0$ imply that $\lambda > 0$. At the same time, the
second root of $p_\lambda(x)$, $\frac{\lambda - 1}{\lambda}$, must lies
in the interval $(\tau_-,\tau_+)$. Hence the equivalent condition for
$\tM$ being positive definite is that
\begin{equation} \label{eq:tM-psd-cond}
\tau_- < \frac{\lambda - 1}{\lambda} < \tau_+ \Leftrightarrow
\left\{
    \begin{array}{ll}
        \lambda > \frac{1}{1-\tau_-}\\
        \lambda < \frac{1}{1-\tau_+} \quad \text{if } \tau_+ < 1
    \end{array}
\right..
\end{equation}
We now move on to the second concern, the condition number of $\tM$
comparing to that of $M^2$. Using the notations above, the condition
number of $M^2$ is
\begin{equation*}
    \cond{M^2} = \frac{\bimax{L_+^2}{L_-^2}}{\bimin{\tau_+^2}{\tau_-^2}},
\end{equation*}
and the condition number of $\tM$ is
\begin{equation} \label{eq:cond-tM-lambda}
    \condp{\lambda}{\tM} = \frac{ \bimax{p_\lambda(L_+)}{p_\lambda(L_-)}
    }{ \bimin{p_\lambda(\tau_+)}{p_\lambda(\tau_-)} }.
\end{equation}
The optimal choice $\lambda^* = \argmin_{\lambda \text{ satisfies}
\eqref{eq:tM-psd-cond}}\condp{\lambda}{\tM}$ is difficult to determine.
On the other hand, a simple choice,
\begin{equation*}
    \hlambda = \left\{
        \begin{array}{ll}
            \frac{1}{1-\tau_+-\tau_-} \quad & \text{if } \tau_+ + \tau_-
            < 1-\frac{1}{C}\\
            C & \text{if } \tau_+ + \tau_- \geq 1 - \frac{1}{C}
        \end{array}
        \right.
\end{equation*}
leads to
\begin{equation*}
    \bimin{p_{\hlambda}(\tau_+)}{p_{\hlambda}(\tau_-)} = \left\{
        \begin{array}{ll}
            -\hlambda \tau_+ \tau_- \quad & \text{if } \tau_+ + \tau_- <
            1 - \frac{1}{C}\\
            \hlambda \tau_- (\tau_- - 1 + \frac{1}{C}) & \text{if }
            \tau_+ + \tau_- \geq 1 - \frac{1}{C}
        \end{array}
        \right.,
\end{equation*}
where $C$ is a sufficiently large constant. When $\abs{\tau_-}$ is
orders of magnitude larger than $\tau_+$ and $\bimax{L_+}{L_-} \gg
\abs{\tau_-}$, the condition number $\condp{\hlambda}{\tM}$ is roughly
$\frac{\abs{\tau_-}}{\tau_+}$ times smaller than $\cond{M^2}$. More
extreme example is that when $\abs{\tau_-} \sim \abs{L_-} > L_+
\gg \tau_+$, the condition number of $\tM$ is roughly constant,
$\condp{\hlambda}{\tM} = \order{1}$, whereas the condition number
$\cond{M^2}$ could be gigantic if the ratio $\abs{L_-}/\tau_+$
is gigantic. Figure~\ref{fig:Poly} shows the comparison between
$p_\lambda(M)$ and $M^2$. The largest eigenvalue of $M^2$ is obviously
larger than that of $p_\lambda(M)$, and the smallest eigenvalue of $M^2$
is also smaller than that of $p_\lambda(M)$ (shown in the zoom-in
subfigure). Therefore, in this case, the condition number of $M^2$
is much larger than that of $p_\lambda(M)$. However, when we swap
the position of $\tau_+$ and $\tau_-$, e.g., $\tau_+$ is orders of
magnitude larger than $\abs{\tau_-}$ and $\bimax{L_+}{L_-} \gg \tau_+$,
the condition number $\condp{\hlambda}{\tM}$ could be of the same order
as $\cond{M^2}$ if $\abs{\tau_-} \sim 1$. The limitation comes from
the restricted expression of $p_\lambda(x)$, where the second root must
be smaller than one if $\lambda > 0$.

\begin{figure}[htp]
    \centering
    \includegraphics[width=0.55\textwidth]{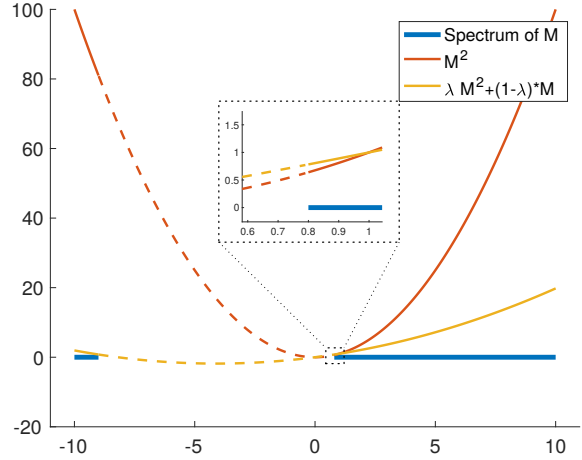}
    \caption{Two quadratic polynomials act on the spectrum of a matrix $M$
    with $L_-=-10$, $\tau_-=-9$, $\tau_+=0.8$ and $L_+=10$.}
    \label{fig:Poly}
\end{figure}

In summary of the above analysis, we observe that the quadratic
polynomial of the matrix, $p_\lambda(M)$ with a careful choice of
$\lambda$, turns $M$ into a positive definite matrix. For a certain
class of matrices, $p_\lambda(M)$ is much well-conditioned than the
traditional technique $M^2$, which is favorable for the later
iterative method. Furthermore, according to the choice of $\hlambda$,
the improvement of the condition number is more significant when
$M$'s close-to-zero eigenvalues are tiled around zero.

\medskip 

\noindent \textbf{Generic quadratic polynomial preconditioning.}
Inspired by the above analysis, we propose a more generic quadratic
polynomial $p_\delta(x) = x^2 - \delta x$ for preconditioning, where
$\delta$ is a parameter playing the similar role as $\lambda$. By
abuse of notation, $\tM = p_\delta(M)$.  Similar as before, $\tM$ and
$M$ share the same eigenvectors and the eigenvalues of $\tM$ are
$p_\delta(m_i) = m_i^2 - \delta m_i$. Since $\delta$ is the second
root of $p_\delta(x)$, $\tau_- < \delta < \tau_+$ guarantees the
positivity of $\tM$. The definition of the condition number of $\tM$
is as \eqref{eq:cond-tM-lambda},
\begin{equation} \label{eq:cond-tM-delta}
    \condp{\delta}{\tM} = \frac{ \bimax{p_\delta(L_+)}{p_\delta(L_-)}
    }{ \bimin{p_\delta(\tau_+)}{p_\delta(\tau_-)} }.
\end{equation}
The optimal choice of $\delta$, $\delta^* = \argmin_{\tau_- < \delta <
\tau_+} \condp{\delta}{\tM}$, is difficult to determine. We adopt the
simple choice $\hdelta = \tau_+ + \tau_-$, leading to
\begin{equation*}
    \bimin{p_{\hdelta}(\tau_+)}{p_{\hdelta}(\tau_-)} = -\tau_- \tau_+.
\end{equation*}
The condition number $\condp{\hdelta}{\tM}$ has similar behavior as
$\condp{\hlambda}{\tM}$ when $\tau_-$ is away from zero and $\tau_+$
is close to zero. Different behavior appears when $\tau_-$ is closer
to zero than $\tau_+$. When $\tau_+$ is orders of magnitude larger than
$\abs{\tau_-}$ and $\bimax{L_+}{L_-} \gg \tau_+$, the condition number
$\condp{\hdelta}{\tM}$ is $\frac{\tau_+}{\abs{\tau_-}}$ times smaller
than $\cond{M^2}$. Therefore $p_\delta(x)$ has broader applicable range
than $p_\lambda(x)$.~\footnote{The behavior of $p_\delta(x)$ can be
achieved by combining $p_\lambda(x)$ and $p_\lambda(-x)$. The choice
of $p_\lambda(x)$ or $p_\lambda(-x)$ depends on spectrum property of
$M$. The resulting numerical method is however more complicated than
using $p_\delta(x)$ alone.} In Figure~\ref{fig:Cond}, we demonstrate the
advantage of $p_{\hdelta}(M)$ over $p_{\hlambda}(M)$ for some matrix $M$.

\begin{figure}[htp]
    \begin{subfigure}[t]{0.48\textwidth}
        \centering
        \includegraphics[width=\textwidth]{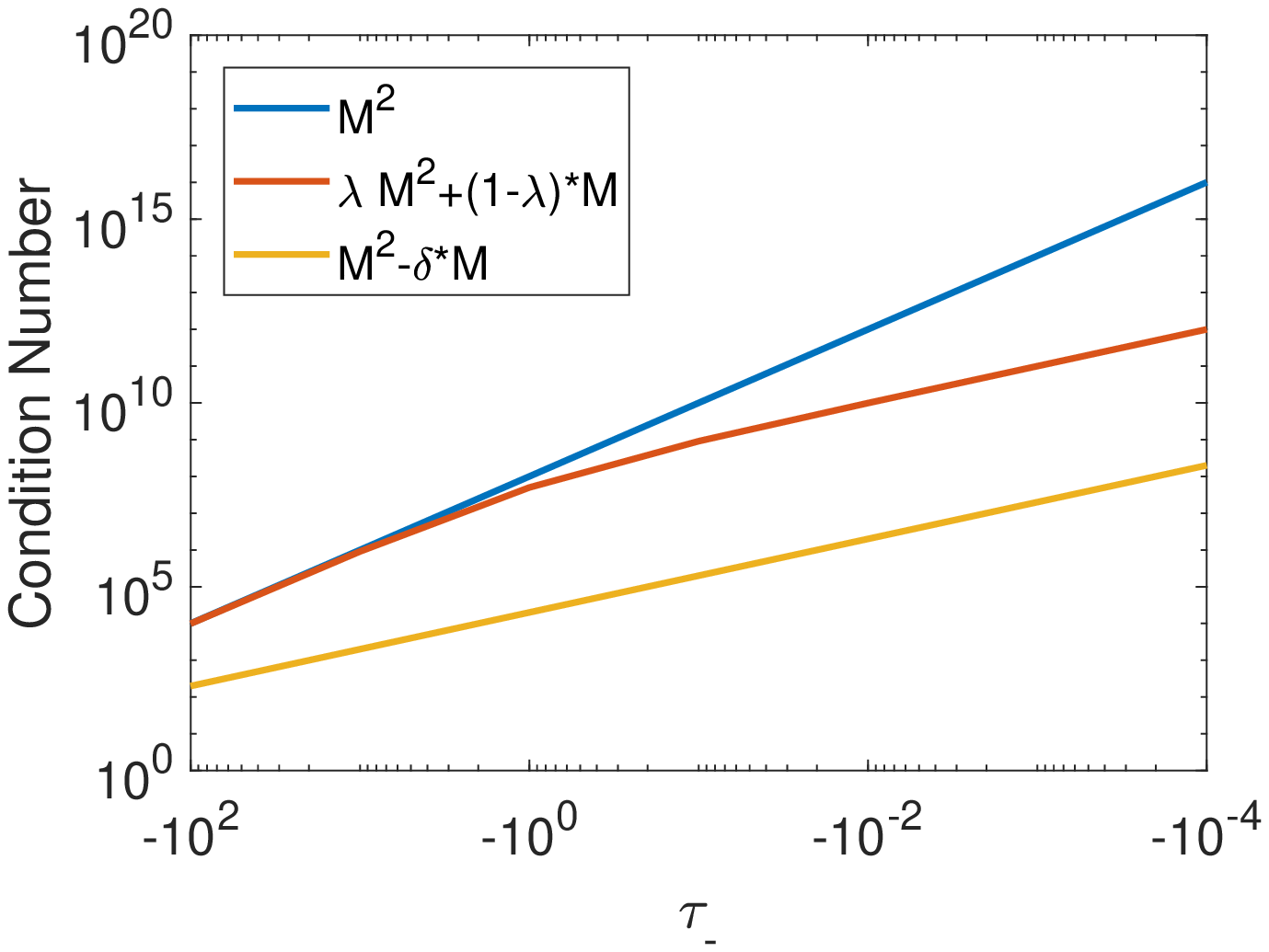}
        \caption{}
        \label{fig:CondTaum}
    \end{subfigure}
    ~
    \begin{subfigure}[t]{0.48\textwidth}
        \includegraphics[width=\textwidth]{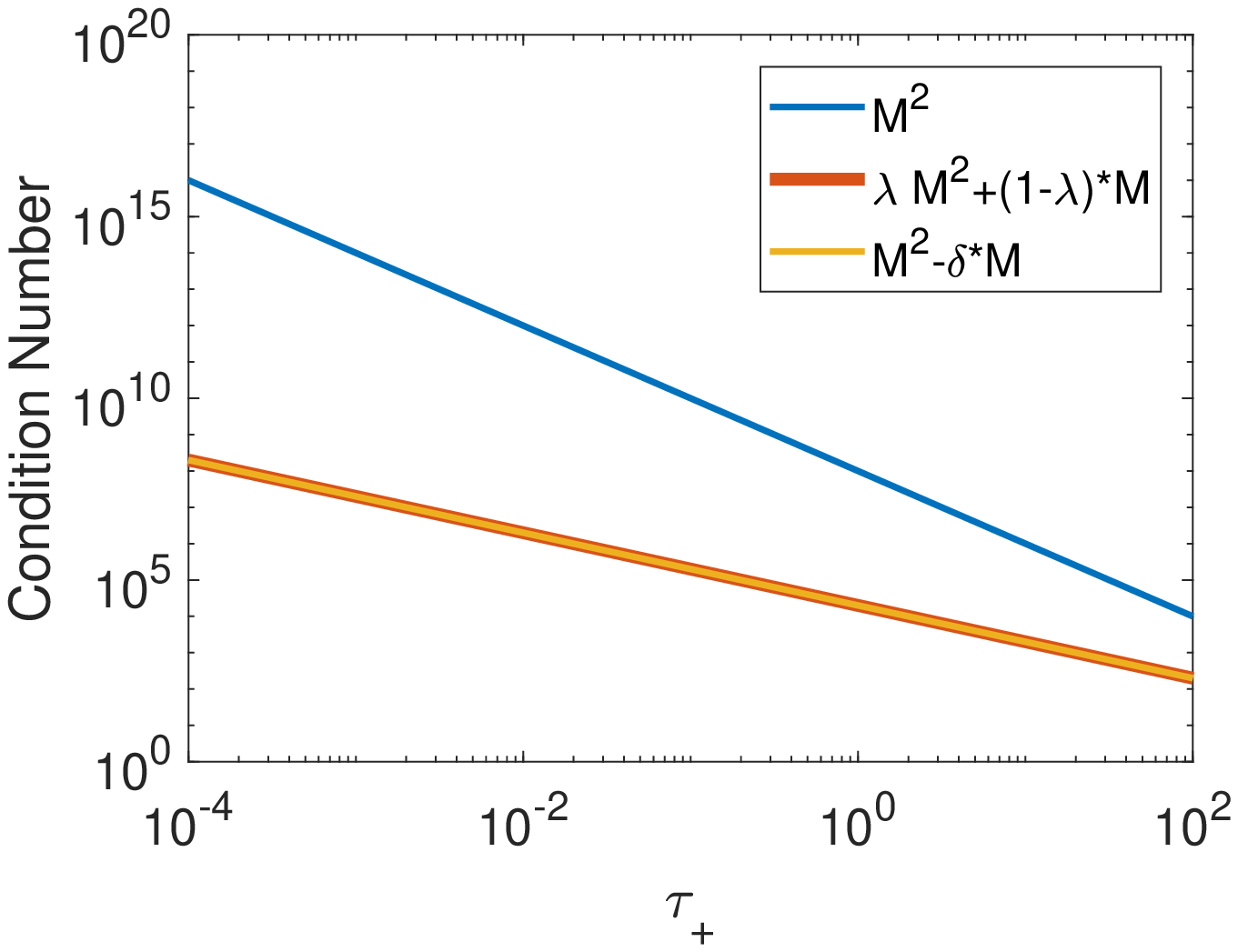}
        \caption{}
        \label{fig:CondTaup}
    \end{subfigure}
    \caption{Condition number of a matrix $M$ with varying asymmetric
    spectrum. In both (a) and (b), the matrix $M$ has fixed $L_+ = 10^4$,
    $L_-=-10^4$. (a) fixes the smallest positive eigenvalue $\tau_+=10^4-1$
    and varies $\tau_-$; (b) fixes the largest negative eigenvalue
    $\tau_-=-10^4+1$ and varies $\tau_+$.} \label{fig:Cond}
\end{figure}

\begin{remark}
    Both $p_{\hdelta}(M)$ and $p_{\hlambda}(M)$ take advantage of the
    asymmetry of the spectrum of $M$. When the spectrum of $M$ is
    symmetric around the origin, i.e., $\tau_+ = \tau_-$ and $L_+ =
    L_-$, the choice of either $p_{\hdelta}(M)$ or $p_{\hlambda}(M)$
    falls back to $M^2$, which has the same condition number as $M^*M$.
\end{remark}

\subsection{BDMC Iterative method}
\label{sec:BDMC-Iter}

In terms of matrix $M$ as in \eqref{eq:BDMC}, the two-step iterative
method in \cite{Prokofev2007} can be written as 
\begin{equation} \label{eq:BDMC-Steps}
    \begin{array}{ll}
        \text{Step 1:} & \ket{\tf_{k+1}} = \ket{b} + (I-M) \ket{f_k}\\
        \text{Step 2:} & \ket{f_{k+1}} = \cfrac{\sum_{j=1}^{k+1} j^t
    \ket{\tf_j}}{\sum_{j=1}^{k+1}j^t}
    \end{array}
\end{equation}
where $t > -1$ is a fixed parameter.~\footnote{The notations have been
changed from \cite{Prokofev2007} as $\alpha \rightarrow t$ and $n
\rightarrow k$ to avoid notation conflicts.} Step 1 is a fixed-point iteration
and Step 2 is a special {N{\o}rlund} mean with sequence $\{j^t\}$. Let
$S_k = \sum_{j=1}^{k}j^t$. We could merge two steps into a single step,
\begin{equation} \label{eq:BDMC-Iter-Init}
    \ket{f_{k+1}} = \frac{1}{S_{k+1}}\ket{f_k} +
    \frac{(k+1)^t}{S_{k+1}}\ket{\tf_{k+1}} = \ket{f_k} -
    \frac{(k+1)^t}{S_{k+1}}\left( M\ket{f_k} - \ket{b} \right).
\end{equation}
For {Hermitian} positive definite matrix $M$, \eqref{eq:BDMC-Iter-Init}
is a gradient descent method for the objective function $\expval{M}{f} -
\bra{b}\ket{f}$ with special stepsize $\alpha_k =
\frac{(k+1)^t}{S_{k+1}}$. Such a stepsize asymptotically behaves as
\begin{equation} \label{eq:stepsize}
    \alpha_k \asymp \frac{t+1}{k+1} \quad (k \rightarrow +\infty).
\end{equation}
In fact, the simpler stepsize choice as \eqref{eq:stepsize} is widely
used in the stochastic gradient descent literature. We would denote
$\beta = t + 1$ in the following note.  The convergence analysis
of the iterative method,
\begin{equation} \label{eq:BDMC-Iter}
    \ket{f_{k+1}} = \ket{f_k} - \frac{\beta}{k+1} \left( M\ket{f_k} -
    \ket{b} \right)
\end{equation}
for both deterministic and stochastic $M$ are listed in the next section.

\subsection{Convergence Analysis}

Let us now turn to the convergence analysis of iterative algorithms
\eqref{eq:BDMC-Iter-Init} and \eqref{eq:BDMC-Iter}. Compared the two, the
analysis of \eqref{eq:BDMC-Iter} would be cleaner due to its simpler
choice of stepsize.  In \cite{Prokofev2007}, {Prokof'ev} and {Svistunov}
provide asymptotic behavior for $\ket{\delta_k} = \ket{f_k} - \ket{f^*}$,
which is the difference between step $k$ result $\ket{f_k}$ and the
underlying truth $\ket{f^*} = M^{-1}\ket{b}$. When $k$ approaches
$+\infty$, $\ket{\delta_k}$ behaves as
\begin{equation} \label{eq:BDMC-asymp}
    \ket{\delta_k} \asymp e^{-\beta M \log{k}} \ket{\delta_1},
\end{equation}
where $\ket{\delta_1} = \ket{f_1} - \ket{f^*}$ and $\ket{f_1}$ is the
initial guess. Since $M$ is a positive definite matrix and $\beta > 0$,
$\ket{\delta_k} \rightarrow \ket{0}$ as $k \rightarrow +\infty$. The
same asymptotic analysis holds for stepsize \eqref{eq:stepsize}
as well.  According to \eqref{eq:BDMC-asymp}, the slowest converging
component behaves as $e^{-\beta \tau_+ \log{k}}$, where $\tau_+$ is the
smallest eigenvalue of $M$. Careful study of the contraction property
of the iterative method shows that either $\beta$ or the number of
non-contraction steps is related to the largest eigenvalue $L_+$ of
$M$. Overall, the smaller condition number of $M$ leads to the faster
convergence.

\begin{remark}
  Based on the above asymptotics, it was suggested in
  \cite{Prokofev2007} the choice of very large $\beta$. In that case,
  for sufficiently large $k$, the asymptotic rate
  \eqref{eq:BDMC-asymp} is achieved. This is however only part of the
  story, as for the iterative method, we are not just interested in
  asymptotic convergence, the actual decay of error after finite
  number of steps is more important. Indeed as we will see in
  Corollary~\ref{corollary}, the hidden prefactor in
  \eqref{eq:BDMC-asymp} depends on $\beta$. In particular, the
  asymptotic analysis fails if $t$ is set as $+\infty$ in the
  iterative method \eqref{eq:BDMC-Iter-Init}. 
  
  Moreover, the asymptotic analysis in \eqref{eq:BDMC-asymp} only
  holds for noise-free matrix $M$. In the stochastic setting, i.e.,
  each evaluation of $M\ket{f}$ involves a stochastic error, we will
  see in Theorem~\ref{thm} that the expected error is dominated by the
  stochastic error part. The choice of large $\beta$ does not impact
  the convergence rate of the iterative method but enlarges the
  prefactor. Therefore, choosing large $\beta$ in the stochastic
  setting actually has negative influence on the convergence.
\end{remark}

The previous asymptotic analysis holds for deterministic matrix $M$. For
stochastic matrix vector multiplication, the analysis is carried out
with assumptions on the bias and variance, we present one possible
convergence result below, following Theorem 4.7 in the review article
\cite{Bottou2016}.
 
Let $G(\ket{f}) = \expval{M}{f} -
\bra{b}\ket{f}$ for Hermitian positive definite matrix $M$. The gradient of
$G(\ket{f})$ is $\grad{G(\ket{f})} = M\ket{f} - \ket{b}$. Hence, both
\eqref{eq:BDMC-Iter-Init} and \eqref{eq:BDMC-Iter} are gradient descent
methods applied to $G(\ket{f})$. In order to distinguish between the
deterministic gradient and stochastic gradient, we denote the stochastic one
as $g(\ket{f},\xi) = M_\xi \ket{f} - \ket{b}$, where $\xi$ is a random
variable. 

\begin{assumption} \label{assump}
    The objective function and stochastic gradient method as
    \eqref{eq:BDMC-Iter} satisfies the following:
    \begin{enumerate}
        \item There exist scalars $\mu_G \geq \mu > 0$ such that, for all
            $k \in \bbN$,
            \begin{equation*}
                \begin{split}
                    & \grad{G(\ket{f_k})}^*\expect{g(\ket{f_k},\xi_k)} \geq
                    \mu \norm{\grad{G(\ket{f_k})}}_2^2,\quad \text{and}\\
                    & \norm{\expect{g(\ket{f_k},\xi_k)}}_2 \leq \mu_G
                    \norm{\grad{G(\ket{f_k})}}_2.
                \end{split}
            \end{equation*}
        \item There exist scalars $W \geq 0$ and $W_V \geq 0$ such that, for
            all $k \in \bbN$,
            \begin{equation*}
                \varance{g(\ket{f_k},\xi_k)} =
                \expect{\norm{g(\ket{f_k},\xi_k)}_2^2} -
                \norm{\expect{g(\ket{f_k},\xi_k)}}_2^2 \leq W + W_V
                \norm{\grad{G(\ket{f_k})}}_2^2.
            \end{equation*}
    \end{enumerate}
\end{assumption}

Assumption~\ref{assump} follows Assumption 4.3 in
\cite{Bottou2016}. When $g(\ket{f_k},\xi_k)$ is an unbiased estimator
of $\grad{G(\ket{f_k})}$, both $\mu_G$ and $\mu$ are one. We recall
the notations $\tau = \tau_+$ and $L = L_+$ as the smallest and
largest eigenvalue of $M$ respectively. Moreover, we denote
$W_G = W_V + \mu_G^2$.

\begin{theorem} \label{thm}
    Under Assumption~\ref{assump}, suppose that the stochastic gradient
    method is run with a stepsize sequence such that for all $k \in \bbN$,
    \begin{equation*}
        \alpha_k = \frac{\beta}{\gamma + k}, \quad \beta > \frac{1}{\tau
        \mu}, \quad \text{and } \gamma > 0 \quad \text{such that } \alpha_1
        \leq \frac{\mu}{LW_G}.
    \end{equation*}
    Then, for all $k \in \bbN$, the expected optimality gap satisfies
    \begin{equation} \label{eq:thm}
        \expect{G(\ket{f_k}) - G_*} \leq \frac{1}{(\gamma + k)^{\beta \tau
        \mu}} \left(G(\ket{f_1}) - G_*\right) (\gamma + 1)^{\beta \tau \mu}
        + \frac{1}{\gamma + k} \frac{ \beta^2 L W }{\beta \tau \mu - 1 }
    \end{equation}
    where $G_* = \inf_{\ket{f}}G(\ket{f})$.
\end{theorem}

\begin{proof}
    Starting from (4.23) in \cite{Bottou2016}, we have,
    \begin{equation} \label{eq:proof}
        \expect{G(\ket{f_{k+1}})} - G_* \leq (1 - \alpha_k \tau \mu) \left(
        \expect{G(\ket{f_k})} - G_* \right) + \frac{1}{2} \alpha_k^2 L W.
    \end{equation}
    Equation~\eqref{eq:thm} holds for $k=1$. Then assuming \eqref{eq:thm} is
    true for $k$, it follows from \eqref{eq:proof} that, ($\hk = \gamma +
    k$),
    \begin{equation*}
        \begin{split}
            \expect{G_{k+1}} - G_* \leq & \left(1 - \frac{\beta \tau
            \mu}{\hk}\right) \hk{}^{-\beta \tau \mu} \left( G(\ket{f_1}) -
            G_*\right) (\gamma + 1)^{\beta \tau \mu}\\
            & + \left( 1 - \frac{\beta \tau \mu}{\hk} \right) \frac{1}{\hk}
            \frac{\beta^2 L W}{\beta \tau \mu - 1} + \frac{1}{2}
            \frac{\beta^2}{\hk{}^2} L W\\
            \leq & \frac{1}{(\hk + 1)^{\beta \tau \mu}} \left( G(\ket{f_1})
            - G_* \right) (\gamma + 1)^{\beta \tau \mu} + \frac{1}{\hk + 1}
            \frac{\beta^2 L W}{\beta \tau \mu -1}
        \end{split}
    \end{equation*}
    where the last inequality dues to Taylor expansion of $(\hk + 1)^{-\beta
    \tau \mu}$ at $\hk$ and $\frac{\hk - 1}{\hk{}^2} < \frac{1}{\hk + 1}$.
\end{proof}

Theorem~\ref{thm} now split the bound into the convergence of the initial
error and stochastic error. Due to the assumption $\beta > \frac{1}{\tau
\mu}$, the expected optimality gap is dominated by the stochastic error,
which behaves as $O(\frac{1}{\gamma + k})$. At the same time, both the
prefactor $\frac{\beta^2 L W}{\beta \tau \mu -1}$ and the parameter of
initial step size $\gamma \geq \frac{L W_G}{\tau \mu^2} -1$ relies on the
condition number of $M$, i.e., $\frac{L}{\tau}$. Therefore, the smaller
condition number of $M$ leads to the faster convergence in stochastic
gradient descent method. A direct corollary can be derived for
non-stochastic gradient descent method, where $\mu = \mu_G = 1, W = 0, W_G
= 1$.

\begin{corollary} \label{corollary}
    Suppose the gradient descent method is run with a stepsize sequence such
    that, for all $k \in \bbN$,
    \begin{equation*}
        \alpha_k = \frac{\beta}{\gamma + k}, \quad \text{and } \gamma > 0
        \quad \text{such that } \alpha_1 \leq \frac{1}{L}.
    \end{equation*}
    Then, for all $k \in \bbN$, the optimality gap satisfies
    \begin{equation*}
        G(\ket{f_k}) - G_* \leq \frac{1}{(\gamma + k)^{\beta \tau}} \left(
        G(\ket{f_1}) - G_* \right) (\gamma + 1)^{\beta \tau}.
    \end{equation*}
\end{corollary}

Corollary~\ref{corollary} coincides with \eqref{eq:BDMC-asymp}.  The
impact of the condition number of $M$ to the convergence is more
significant in the non-stochastic gradient descent method. The
constant $(\gamma+1)^{\beta \tau}$ and parameter $\gamma$ are
influenced by the condition number in a similar way as the stochastic
one. The rate of the convergence of the non-stochastic gradient
descent method is also impacted by the smallest eigenvalue of $M$. In
general, the larger of $\tau$ leads to faster convergence rate of the
gradient descent method. Such an argument agrees with the asymptotic
analysis of \eqref{eq:BDMC-asymp}.


\section{Alternative stochastic iterative methods}
\label{sec:SGDs}

Gradient descent method with stochastic gradient is widely explored in
many areas. Especially, machine learning researchers established many
variant stochastic gradient descent methods to minimize the loss
function in a big data setting. The raise of deep learning further
accelerates the development of stochastic gradient descent methods. In
these context, many loss functions are non-convex
functions. Stochastic gradient descent methods, without accessing full
gradient at each step, samples a few components of the full gradient
and move along the sampled gradient direction. This strategy reduces
the computational cost each step and potentially avoids many local
minima. The problem BDMC addresses, in contrast with deep learning,
has a quadratic convex objective function. On the other hand, the gradient of
the objective function in BDMC can only be accessed via a \MC{}
procedure. The underlying true gradient is unknown. 
Here we would like to compare BDMC with a few well-established
stochastic gradient descent methods from machine learning
literature. The cross-fertilization between machine learning and
computational physics is rather natural since both are dealing with
high dimensional problems. In particular, diagrammatic summation
methods like BDMC could potentially benefit from other stochastic
iterative methods to either improve convergence or allow larger \MC{}
error.

\subsection{Heavy ball method}
\label{sec:HB}

The heavy ball method~\cite{Polyak1964} adds a momentum term to the gradient
descent method,
\begin{equation} \label{eq:HeavyBall}
    \ket{f_{k+1}} = \ket{f_k} - \alpha ( M \ket{f_k} - \ket{b} ) +
    \beta \left( \ket{f_k} - \ket{f_{k-1}} \right),
\end{equation}
where $\alpha$ is the stepsize and $\beta$ is the weight for
momentum. Heavy ball method actually has the same convergence rate as the
gradient descent method. Unlike gradient descent method which depends on the
condition number of $M$, the heavy ball method depends on the square root of
the condition number of $M$. Such a property is attractive when $M$ is
ill-conditioned. However, due to the momentum, the heavy ball method is not
strictly decreasing, i.e., $G(\ket{f_{k+1}})$ is not necessarily smaller than
$G(\ket{f_k})$.

\subsection{Stochastic {B}arzilai-{B}orwein method}
\label{sec:sBB}

In 1988, {Barzilai} and {Borwein} proposed a two-point step size gradient
method~\cite{Barzilai1988}, which is inspired by the secant equation
underlying quasi-Newton methods. The stochastic version of the Barzilai and
Borwein method (sBB)~\cite{Tan2016}, instead of updating the stepsize every
iteration with the difference of stochastic gradients, updates the stepsize
every $m$ iteration with the difference of aggregated gradients. The
detailed iteration is as follows,
\begin{equation} \label{eq:BB-stepsize}
    \begin{split}
        & \alpha_k = \left\{
            \begin{array}{ll}
                \alpha_{k-1} & \text{if } k \not\equiv 0 \mod m\\
                \frac{1}{m}\frac{ \bra{\Delta_k^m f}\ket{\Delta_k^m f} }
                { \bra{\Delta_k^m f}\ket{\Delta_k^m g} }
                & \text{if } k \equiv 0 \mod m\\
            \end{array}    
        \right.\\
        & \ket{f_{k+1}} = \ket{f_{k}} - \alpha_k \grad{G(\ket{f_{k}})}\\
        & \ket{g_{k+1}} = \beta \grad{G(\ket{f_{k}})} + (1-\beta) \ket{g_k}
    \end{split}
\end{equation}
where $\ket{\Delta_k^m f} = \ket{f_{k}} - \ket{f_{k-m}}$ and
$\ket{\Delta_k g} = \ket{g_{k}} - \ket{g_{k-m}}$. $\beta$ is the weight for
momentum and $m$ is the updating frequency. Notice that in \cite{Tan2016}, a
smoothing technique is suggested for the stepsize, which is a technique
enforce diminishing stepsize. According to our tests, this technique is
crucial for convergence when the gradient is noisy. Therefore, our
implementation of sBB adopts the smoothing technique.

\subsection{AdaGrad method}
\label{sec:AdaGrad}

AdaGrad method~\cite{Duchi2011} approximates the Hessian of $G$ by a
diagonal matrix, and sets the stepsize in the gradient method as the
inverse of the diagonal approximation. The iteration of AdaGrad method is
\begin{equation} \label{eq:AdaGrad}
    \ket{f_{k+1}} = \ket{f_k} - \alpha
    \diag{\Gamma_k}^{-1/2}\grad{G(\ket{f_k})},
\end{equation}
where $\Gamma_k = \sum_{i=1}^k \left( \grad{G(\ket{f_k})} \right)^2$ and
the square is an entry-wise operation, $\alpha$ is a parameter. Since
the computational costs for both the inverse of the square root of a
diagonal matrix and the diagonal matrix vector multiplication are the
same as generating the gradient vector.  Therefore, such a choice of
stepsize only increase the computational cost by a small constant,
while the convergence could be accelerate for stochastic gradients.

\subsection{ADAM method}
\label{sec:ADAM}

ADAM method~\cite{Kingma2015} is a upgraded version of AdaGrad method.
Instead of using raw gradient $\grad{G(\ket{f_k})}$ and $\Gamma_k$
as in \eqref{eq:AdaGrad}, ADAM method adds momentum parts for both and
correct the biases. The calculation of its stepsize is more complicated
than all pre-mentioned methods. We summarize the calculation as follows,
\begin{equation}
    \begin{split}
        & \ket{g_{k+1}} = \grad{G(\ket{f_{k}})}\\
        & \ket{m_{k+1}} = \beta_1 \ket{m_k} + (1-\beta_1) \ket{g_{k+1}}
        ~\text{[Update biased first moment]}\\
        & \ket{v_{k+1}} = \beta_2 \ket{v_k} + (1-\beta_2) \ket{g_{k+1}}^2
        ~\text{[Update biased second moment]}\\
        & \ket{\hm{}_{k+1}} = \ket{m_{k+1}}/(1 - \beta_1^{k+1})
        ~\text{[Correct biased first moment]}\\
        & \ket{\hv{}_{k+1}} = \ket{v_{k+1}}/(1 - \beta_2^{k+1})
        ~\text{[Correct biased second moment]}\\
        & \ket{f_{k+1}} = \ket{f_{k}} - \alpha
        \hm{}_{k+1}/(\sqrt{\hv{}_{k+1}} + \epsilon)\\
    \end{split}
\end{equation}
with initial first moment vector $\ket{m_1} = \ket{0}$ and second
moment vector $\ket{v_1} = \ket{0}$, where $\beta_1$ and $\beta_2$
are weights for the first and second moment respectively, $\alpha$ is
a parameter. These initial first and second moment are crucial for the
bias correction parts. Although the iterative scheme looks much more
complicated than before, the computational cost is about 2-3 times as
much as that of AdaGrad method.


\section{Numerical Results}
\label{sec:NumRes}

This section focus on testing the performances of different gradient
descent methods mentioned in Section~\ref{sec:BDMC-Iter} and
Section~\ref{sec:SGDs}. The power of the preconditioning technique has
been illustrated in Figure~\ref{fig:Cond}, and we would not retest it
here. The algorithms of different gradient descent methods are implemented
in MATLAB 2017b and the results reported here are obtained on a MacBook
Pro with 2.3 GHz Intel Core i7 and 8 GB memory.

For simplicity, we will use the shorten name as in
Table~\ref{tab:sname} instead of the original full name. The accuracy
of all iterative methods is measured against the underlying true
solution $\ket{f^*}$ as,
\begin{equation} \label{eq:relerr}
    e_k^{rel} = \frac{\norm{\ket{f_k} - \ket{f^*}}}{\norm{f^*}},
\end{equation}
where $\ket{f_k}$ is the solution at $k$th step and $e_k^{rel}$ is
called the relative error at $k$th step.

\begin{table}[htp]
    \centering
    \begin{tabular}{ccc}
        \toprule
        Short Name & Full Name & Scheme \\
        \toprule
        GD & Gradient descent method & \\
        BDMC & Bold Diagrammatic \MC{} method & Section~\ref{sec:BDMC-Iter}
        \eqref{eq:BDMC-Steps} \\
        BDMC2 & Bold Diagrammatic \MC{} method & Section~\ref{sec:BDMC-Iter}
        \eqref{eq:BDMC-Iter} \\
        HB & Heavy ball method & Section~\ref{sec:HB} \\
        sBB & Stochastic Barzilai-Borwein method & Section~\ref{sec:sBB} \\
        AdaGrad & Adaptive gradient method &
        Section~\ref{sec:AdaGrad} \\
        ADAM & Adaptive moment estimation method &
        Section~\ref{sec:ADAM} \\
        \bottomrule
    \end{tabular}
    \caption{Name convention for iterative methods.}
    \label{tab:sname}
\end{table}

One example is to simulate the DMC by noisy symmetric positive definite
matrices $M$ of size 100 by 100. We first generate the simulating system as
follows,
\begin{equation}
    M = Q
    \begin{bmatrix}
        1 & & & \\
        & 2 & & \\
        & & \ddots & \\
        & & & 100
    \end{bmatrix}
    Q^*,
\end{equation}
where $Q$ is a random unitary matrix. Then a underlying true solution
$\ket{f^*}$ is generated from normal distribution. The vector $\ket{b}$,
therefore, is the multiplication of $M$ and $\ket{f^*}$, i.e., $\ket{b} = M
\ket{f^*}$. In order to simulate the uncertainty of the DMC, we added noise
to each entry of $M\ket{f}$, i.e.,
\begin{equation}
    g(\ket{f},\xi) = M \ket{f} - \ket{b} + \epsilon \ket{\xi},
\end{equation}
where $\ket{\xi}$ is a random vector with each entry drawn from standard
normal distribution and $\epsilon$ is the noise level.

\begin{table}[htp]
    \centering
    \begin{tabular}{c|c|c|c|cc|cc|c|ccc}
        \toprule
        & GD & BDMC & BDMC2 & \multicolumn{2}{c}{HB} &
        \multicolumn{2}{|c|}{sBB} & AdaGrad & \multicolumn{3}{c}{ADAM} \\
        $\epsilon$ & $\alpha$ & $t$ & $\beta$ & $\alpha$ & $\beta$ & $m$ &
        $\beta$ & $\alpha$ & $\alpha$ & $\beta_1$ & $\beta_2$ \\
        \toprule
        $0.01$ & $  0.005$ & $-0.6$ & $0.4$ & $   0.01$ & $0.6$ & $50$ &
        $0.9$ & $0.9$ & $  0.1$ & $0.99$ & $0.999$ \\
        $ 0.1$ & $  0.001$ & $-0.6$ & $0.4$ & $  0.001$ & $0.6$ & $50$ &
        $0.4$ & $0.4$ & $ 0.01$ & $0.99$ & $0.999$ \\
        $   1$ & $ 0.0005$ & $-0.5$ & $0.5$ & $ 0.0005$ & $0.6$ & $30$ &
        $0.5$ & $0.2$ & $0.001$ & $0.99$ & $0.999$ \\
        $  10$ & $ 0.0003$ & $-0.5$ & $0.5$ & $ 0.0001$ & $0.6$ & $40$ &
        $0.5$ & $0.2$ & $0.002$ & $0.99$ & $0.999$ \\
        $ 100$ & $0.00008$ & $-0.4$ & $0.5$ & $0.00005$ & $0.6$ & $50$ &
        $0.4$ & $0.2$ & $0.002$ & $0.99$ & $0.999$ \\
        \bottomrule
    \end{tabular}
    \caption{``Close-to-optimal'' parameters of iterative methods for
    the first example.}
    \label{tab:opt-paras-ex1}
\end{table}

Each different gradient descent method in Section~\ref{sec:BDMC-Iter} and
Section~\ref{sec:SGDs} has some parameters in common, maximum number of
iteration is $10000$, convergence tolerance is $10^{-6}$, initial guess
$\ket{f_1}$ is a random vector with entry drawn from standard normal
distribution (except that ADAM always starts from $\ket{0}$). Besides
these common parameters, these methods have their own parameters requiring
tuning. For each method, we tried different settings and summarize the
close-to-optimal parameter in Table~\ref{tab:opt-paras-ex1} up to one
significant digits. ``Close-to-optimal'' is in the sense that the averaged
relative error $\frac{\norm{\ket{f_{10000}}-\ket{f^*}}}{\norm{\ket{f^*}}}$
of 10 runs is minimized.

\begin{figure}[htp]
    \centering
    \begin{subfigure}[t]{0.48\textwidth}
        \includegraphics[width=\textwidth]{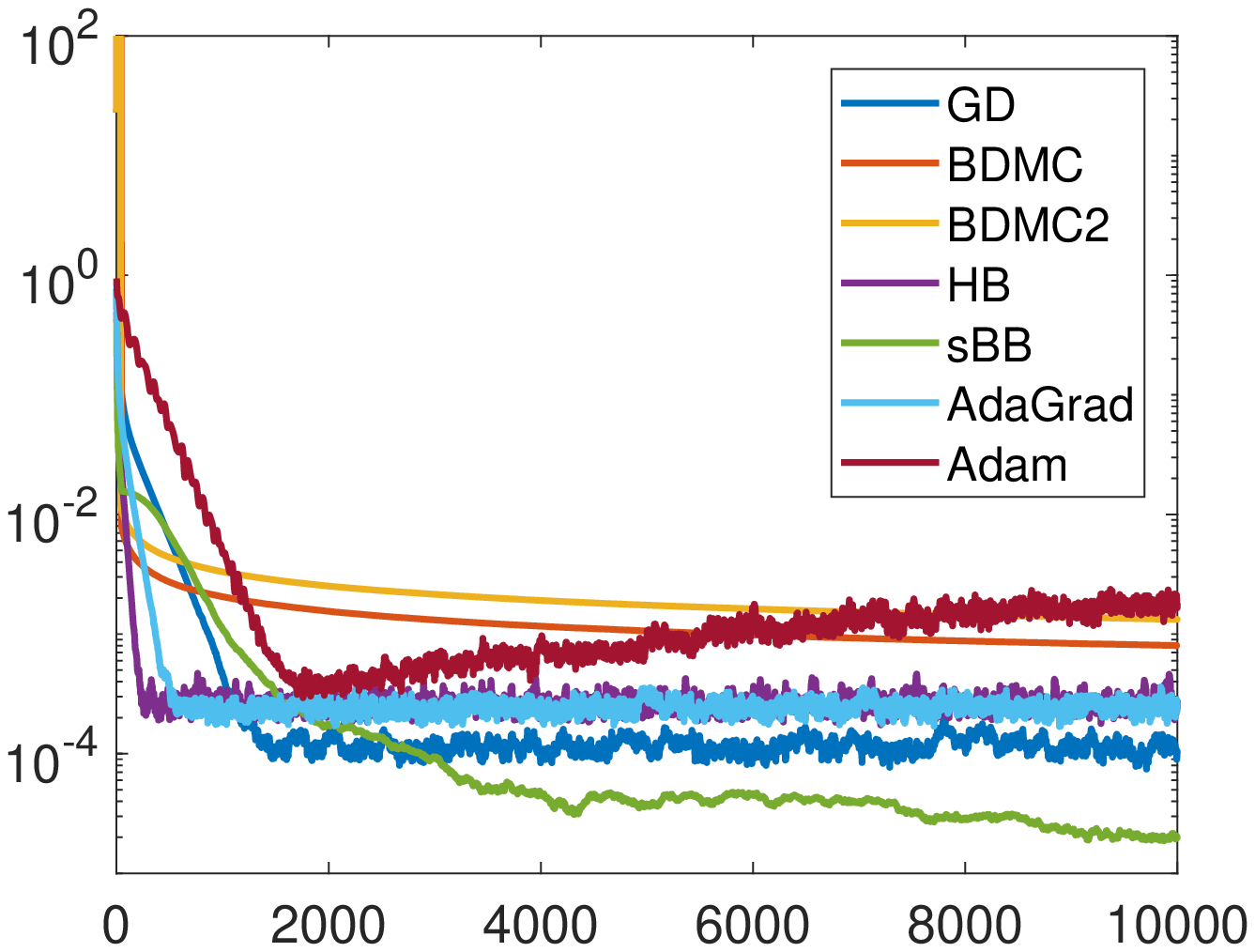}
        \caption{$\epsilon = 0.01$}
    \end{subfigure}
    ~
    \begin{subfigure}[t]{0.48\textwidth}
        \includegraphics[width=\textwidth]{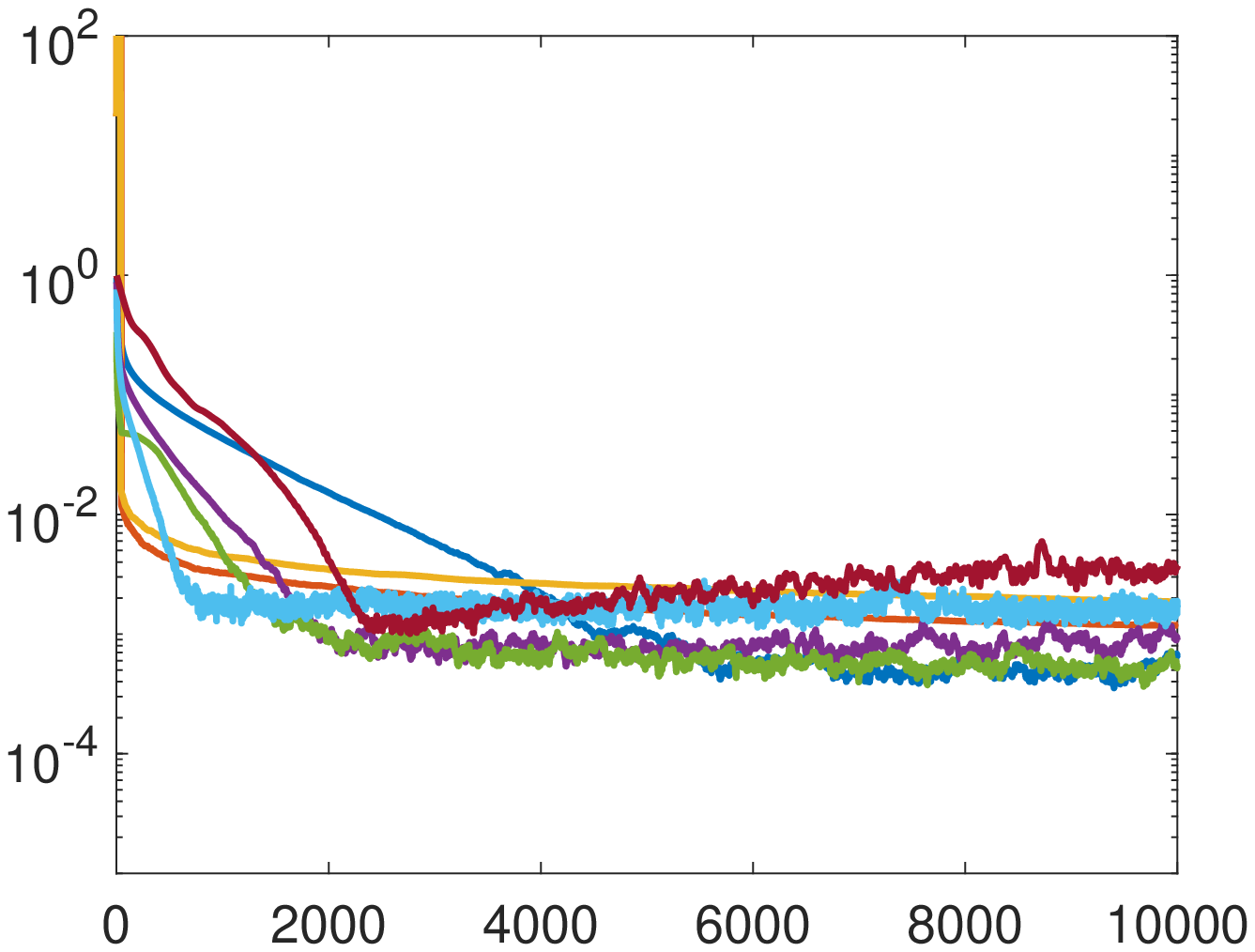}
        \caption{$\epsilon = 0.1$}
    \end{subfigure}
    \\ 
    \begin{subfigure}[t]{0.48\textwidth}
        \includegraphics[width=\textwidth]{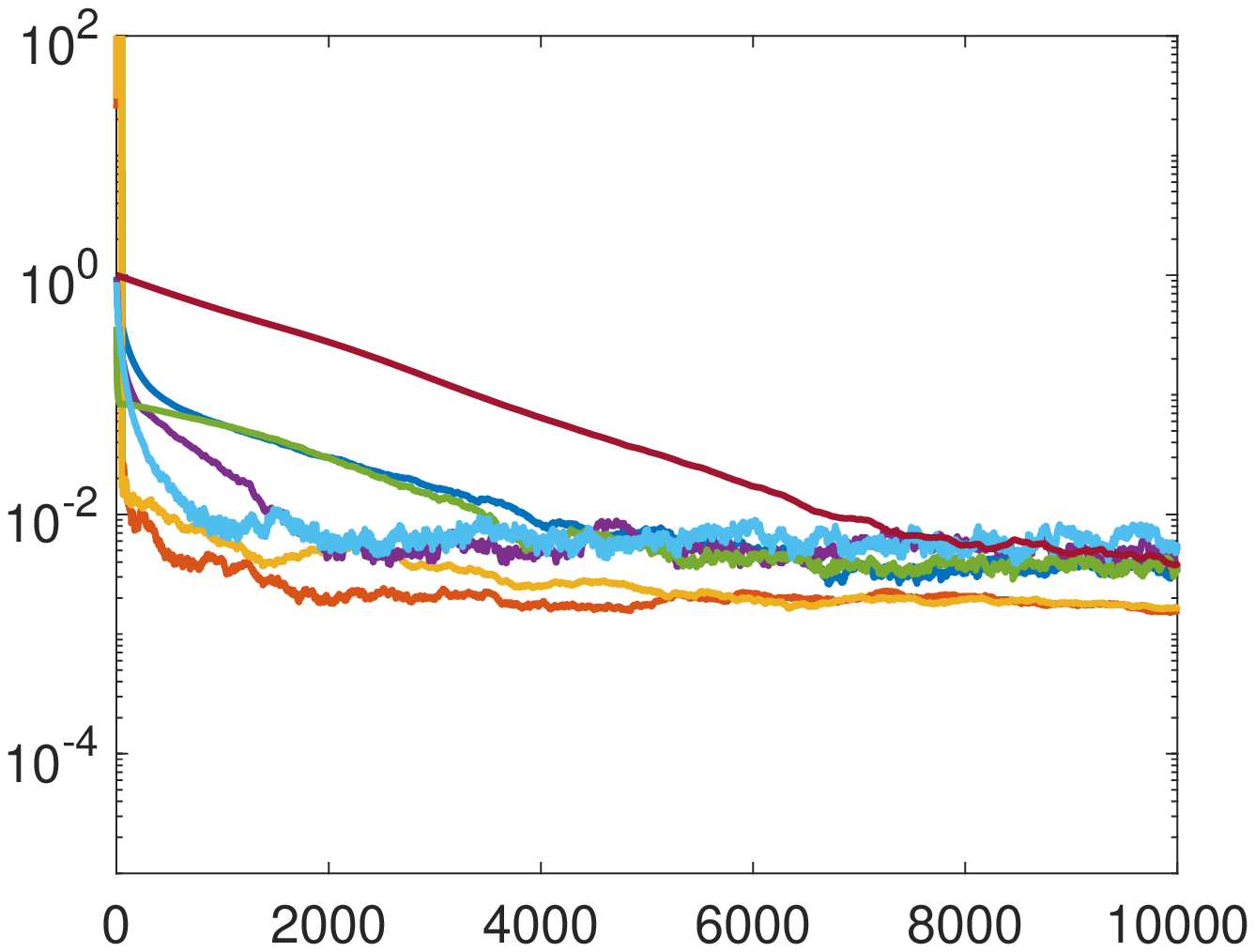}
        \caption{$\epsilon = 1$}
    \end{subfigure}
    ~
    \begin{subfigure}[t]{0.48\textwidth}
        \includegraphics[width=\textwidth]{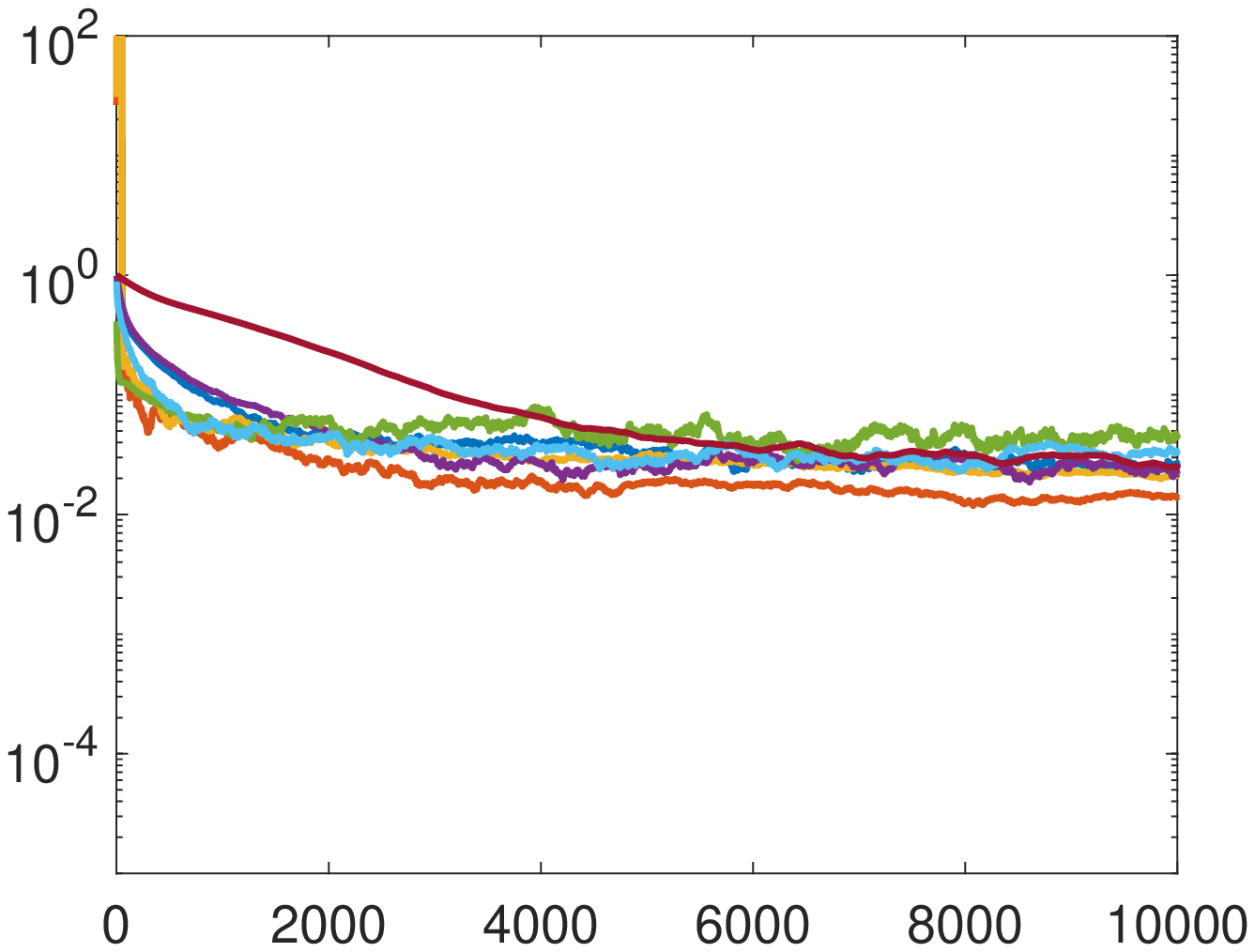}
        \caption{$\epsilon = 10$}
    \end{subfigure}
    \\
    \begin{subfigure}[t]{0.48\textwidth}
        \includegraphics[width=\textwidth]{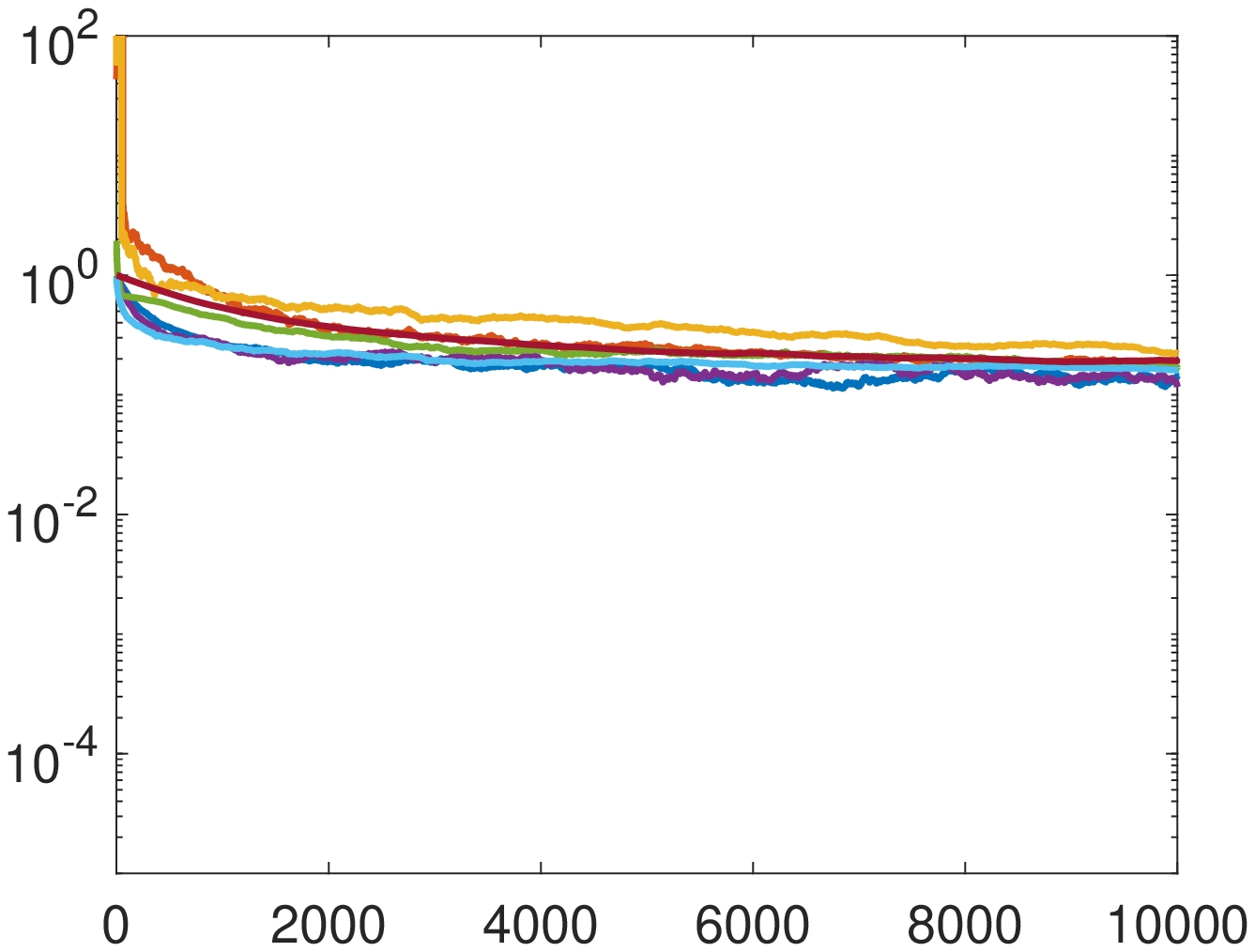}
        \caption{$\epsilon = 100$}
    \end{subfigure}
    \begin{subfigure}[t]{0.48\textwidth}
        \includegraphics[width=\textwidth]{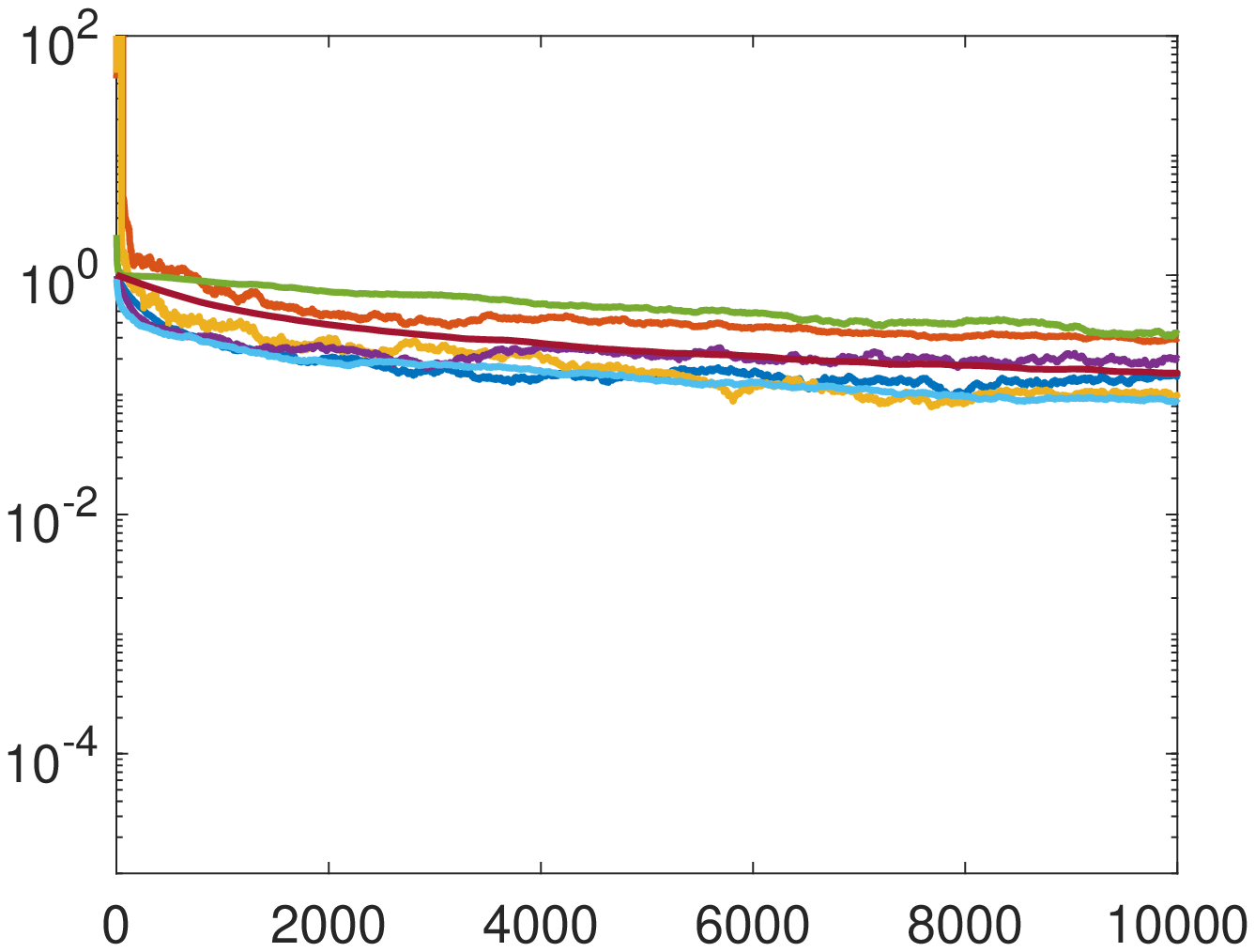}
        \caption{$\epsilon = 100$}
    \end{subfigure}
    \caption{Convergence results for different methods and different
    noise level. For all subfigures, $y$-axis denotes the relative
    error as in \eqref{eq:relerr} and $x$-axis denotes the iteration
    number. The legends for (b) - (f) are the same as that in (a).}
    \label{fig:Iterations-ex1}
\end{figure}

Figure~\ref{fig:Iterations-ex1} illustrates the performance of
different stochastic gradient descent methods with different levels of
noise in the gradient. According to Figure~\ref{fig:Iterations-ex1}
(a) and (b), which correspond to low noise cases, the stochastic
Barzilai-Borwein method outperforms other methods. Although it is not
the best method in the first 2000 iterations, it achieves the best
relative error when the iteration number getting larger. In these low
noise cases, BDMC and BDMC2 do not perform well comparing to other
methods. While, in the medium noise cases (see
Figure~\ref{fig:Iterations-ex1} (c) and (d)), BDMC and BDMC2 are the
best methods among all. They outperform other methods from the very
beginning of the iterations. Based on our numerical tests,
Figure~\ref{fig:Iterations-ex1} (a), (b), (c) and (d) are relatively
robust with respect to different runs of the algorithms. However, as
the noise level being close to the largest eigenvalue of $M$, the
test results are no longer robust. The ranking of the methods shifts
randomly. For example, Figure~\ref{fig:Iterations-ex1} (e) and (f) are
two runs of the methods at the same noise level $\epsilon = 100$. In
(e), BDMC is the worst method, whereas in (f) it is one of the
best. Therefore, conclusion of the ranking of the methods cannot be
made here for the high noise level case.

We also would like to raise one concern about the BDMC and BDMC2
method. For both of them, the relative errors ``blow up'' in the
first few iterations and quickly drop down to a reasonable level. And
the peaks of the ``blow-ups'' could be as large as $10^{10}$ for the
example here. These peaks are also related to the condition number of
the original matrix $M$. The worse of the condition number of $M$,
the higher of the peak. At the same time, the parameter $t$ must be
tuned to close to $-1$ ($\beta$ in BDMC2 be tuned to close to $0$)
to enable the drop-down behavior and achieve convergence.

The last point is about the sensitivity of the parameters. The
parameters of gradient descent method and heavy ball method are the
most sensitive ones, i.e., small change in the parameters would result
huge performance difference. On the other side, the parameters in the
stochastic Barzilai-Borwein method and AdaGrad method are the least
sensitive ones. Although in Table~\ref{tab:opt-paras-ex1}, their
parameters vary a lot, but many other choices of their parameters
actually show similar convergence behavior. Therefore, these two
methods are easier to use in practice.


\section{Conclusion}
\label{sec:Conclusion}

This note provides mathematical understanding of the original BDMC method
in \cite{Prokofev2007}. The two parts in the BDMC method are interpreted
as the preconditioning part and stochastic iterative part. In the
preconditioning part, a quadratic polynomial of the matrix,
$p_{\hlambda}(M)$, turns an indefinite matrix into a positive definite
matrix and the corresponding condition number could potentially be orders
of magnitudes smaller than the traditional preconditioning technique,
$M^2$.  In addition to $p_{\hlambda}(M)$, we propose another quadratic
polynomial $p_{\hdelta}(M)$ which has the same performance on some
matrices and achieves better performance on another big group of matrices.
For the second part, the stochastic iterative part, we rewrite the
original multi-step BDMC method as a gradient descent method with
diminishing step size, \eqref{eq:BDMC-Iter-Init}. Asymptotically, the
complicated stepsize can be replaced by \eqref{eq:BDMC-Iter}. The choices
of both stepsizes behave similar on all numerical examples we have tested.
Due to the DMC procedure involved in the evaluation of the matrix, the
BDMC iterative method is actually a stochastic gradient descent method on
quadratic objective function. Naturally, we introduce a few stochastic
gradient descent methods from machine learning and deep learning, which
are originally designed for non-convex objective functions. All these
stochastic gradient descent methods are tested on a simulated toy example
with different level of noises. In the small noise levels, stochastic
Barzilai-Borwein shows great power over other method.  For medium noise
level close to the smallest eigenvalue of the matrix, BDMC methods (with
two choices of stepsize) outperform other methods. When the noise level is
as large as the largest eigenvalue of the matrix, the conclusion for the
performance of methods is unclear.

At this point, the new preconditioning technique and variant stochastic
gradient descent methods are only tested on simulated matrices with noise.
In the future, we would like to apply all these techniques and methods to
the actual interacting systems using bold diagrammatic \MC{} method.


\medskip 
\noindent {\bf Acknowledgments.}  This work is partially supported by
the National Science Foundation under awards OAC-1450280 and
DMS-1454939. We thank Lexing Ying for interesting discussions regards
the BDMC method.

\bibliographystyle{abbrv} \bibliography{ref}

\begin{thebibliography}{10}

\bibitem{Austin2012}
B.~M. Austin, D.~Y. Zubarev, and W.~A. Lester~Jr.
\newblock Qunatum {M}onte {C}arlo and related approaches.
\newblock {\em Chem. Rev.}, 112:263--288, 2012.

\bibitem{Barzilai1988}
J.~Barzilai and J.~M. Borwein.
\newblock {Two-point step size gradient methods}.
\newblock {\em IMA J. Numer. Anal.}, 8(1):141--148, jan 1988.

\bibitem{Bottou2016}
L.~Bottou, F.~E. Curtis, and J.~Nocedal.
\newblock {Optimization methods for large-scale machine learning}.
\newblock Technical report, jun 2016.

\bibitem{Ceperley2010}
D.~M. Ceperley.
\newblock An overview of quantum {M}onte {C}arlo methods.
\newblock {\em Rev. Mineral. Geochem.}, 71:129--135, 2010.

\bibitem{Duchi2011}
J.~Duchi, E.~Hazan, and Y.~Singer.
\newblock {Adaptive subgradient methods for online learning and stochastic
  optimization}.
\newblock {\em J. Mach. Learn. Res.}, 12:2121--2159, 2011.

\bibitem{FetterWalecka2003}
A.~L. Fetter and J.~D. Walecka.
\newblock {\em Quantum theory of many-particle systems}.
\newblock Dover, 2003.

\bibitem{Foulkes2001}
W.~M.~C. Foulkes, L.~Mitas, R.~J. Needs, and G.~Rajagopal.
\newblock Quantum {M}onte {C}arlo simulations of solids.
\newblock {\em Rev. Mod. Phys.}, 73:33--83, 2001.

\bibitem{Kingma2015}
D.~P. Kingma and J.~Ba.
\newblock {Adam: A method for stochastic optimization}.
\newblock In {\em 3rd Int. Conf. Learn. Represent.}, San Diego, dec 2015.

\bibitem{Kolorenc2011}
J.~Kolorenc and L.~Mitas.
\newblock Applications of quantum {M}onte {C}arlo methods in condensed systems.
\newblock {\em Rep. Prog. Phys.}, 4:026502, 2011.

\bibitem{Kulagin2013}
S.~A. Kulagin, N.~V. Prokof'ev, O.~A. Starykh, B.~V. Svistunov, and C.~N.
  Varney.
\newblock {Bold diagrammatic {Monte} {Carlo} method applied to {Fermionized}
  frustrated spins}.
\newblock {\em Phys. Rev. Lett.}, 110(7):070601, feb 2013.

\bibitem{LeBlanc2015}
J.~P.~F. LeBlanc, A.~E. Antipov, F.~Becca, I.~W. Bulik, G.~K.-L. Chan, C.-M.
  Chung, Y.~Deng, M.~Ferrero, T.~M. Henderson, C.~A. Jim{\'{e}}nez-Hoyos,
  E.~Kozik, X.-W. Liu, A.~J. Millis, N.~V. Prokof'ev, M.~Qin, G.~E. Scuseria,
  H.~Shi, B.~V. Svistunov, L.~F. Tocchio, I.~S. Tupitsyn, S.~R. White,
  S.~Zhang, B.-X. Zheng, Z.~Zhu, and E.~Gull.
\newblock {Solutions of the two-dimensional {Hubbard} model: Benchmarks and
  results from a wide range of numerical algorithms}.
\newblock {\em Phys. Rev. X}, 5:041041, 2015.

\bibitem{Mattuck1992}
R.~D. Mattuck.
\newblock {\em A guide to {F}eynman diagrams in the many-body problem: second
  edition}.
\newblock Dover, 1992.

\bibitem{Paige1975}
C.~C. Paige and M.~A. Saunders.
\newblock {Solution of sparse indefinite systems of linear equations}.
\newblock {\em SIAM J. Numer. Anal.}, 12(4):617--629, sep 1975.

\bibitem{Polyak1964}
B.~T. Polyak.
\newblock {Some methods of speeding up the convergence of iteration methods}.
\newblock {\em USSR Comput. Math. Math. Phys.}, 4(5):1--17, jan 1964.

\bibitem{Prokofev2007}
N.~V. Prokof'ev and B.~V. Svistunov.
\newblock {Bold diagrammatic {Monte} {Carlo} technique: When the sign problem
  is welcome}.
\newblock {\em Phys. Rev. Lett.}, 99(25):250201, 2007.

\bibitem{Prokofev2008}
N.~V. Prokof'ev and B.~V. Svistunov.
\newblock {Bold diagrammatic {Monte} {Carlo}: A generic sign-problem tolerant
  technique for polaron models and possibly interacting many-body problems}.
\newblock {\em Phys. Rev. B}, 77(12):125101, 2008.

\bibitem{Prokofev2008a}
N.~V. Prokof'ev and B.~V. Svistunov.
\newblock {{Fermi}-polaron problem: Diagrammatic {Monte} {Carlo} method for
  divergent sign-alternating series}.
\newblock {\em Phys. Rev. B}, 77(2):020408, jan 2008.

\bibitem{Saad1986}
Y.~Saad and M.~H. Schultz.
\newblock {{GMRES}: A generalized minimal residual algorithm for solving
  nonsymmetric linear systems}.
\newblock {\em SIAM J. Sci. Stat. Comput.}, 7(3):856--869, jul 1986.

\bibitem{Tan2016}
C.~Tan, S.~Ma, Y.-H. Dai, and Y.~Qian.
\newblock {{Barzilai}-{Borwein} step size for stochastic gradient descent}.
\newblock In D.~D. Lee, M.~Sugiyama, U.~V. Luxburg, I.~Guyon, and R.~Garnett,
  editors, {\em Adv. Neural Inf. Process. Syst. 29}, pages 685--693. Curran
  Associates, Inc., 2016.

\bibitem{VanHoucke2008}
K.~Van~Houcke, E.~Kozik, N.~Prokof'ev, and B.~Svistunov.
\newblock Diagrammatic {M}onte {C}arlo.
\newblock In D.~P. Laudan, S.~P. Lewis, and H.~B. Schuttler, editors, {\em
  Computer Simulation Studies in Condesed Matter Physics XXI}. Springer, 2008.

\bibitem{VanHoucke2012}
K.~{Van Houcke}, F.~Werner, E.~Kozik, N.~V. Prokof'ev, B.~V. Svistunov,
  M.~J.~H. Ku, A.~T. Sommer, L.~W. Cheuk, A.~Schirotzek, and M.~W. Zwierlein.
\newblock {{Feynman} diagrams versus {Fermi}-gas {Feynman} emulator}.
\newblock {\em Nat. Phys.}, 8(5):366--370, 2012.

\end{thebibliography}

\end{document}